\documentclass[a4paper,english,numberwithinsect]{eurocg25}

\usepackage{enumerate, amsmath, amsfonts, amssymb, amsthm, mathrsfs, complexity, wasysym, graphics, graphicx, url, hyperref, hypcap, shuffle, xargs, multicol, overpic, pdflscape, multirow, hvfloat, minibox, accents, array, multido, xifthen, xspace, ae, aecompl, blkarray, pifont, mathtools, etoolbox, dsfont, verbatim, stackrel, stmaryrd, xcolor, libertine, dynkin-diagrams, nicefrac, thmtools,svg,csquotes,tcolorbox, multirow, cleveref}
\usepackage[disable]{todonotes}

\newcommand{\robert}[1]{\textcolor{green!70!black}{Robert: #1}}
\newcommand{\helena}[1]{\textcolor{orange!90}{Helena: #1}}

\newcommand{\felix}[1]{\textcolor{blue!90}{Felix: #1}}

\newtheorem{problem}{Problem}
\newtheorem{conjecture}{Conjecture}

\newtheorem*{claim*}{Claim}

\newenvironment{claimproof}{{\bf Proof of the claim:}}{\hfill}

\newcommand{\cT}{\ensuremath{\mathcal{T}}}

\DeclareMathOperator{\Cone}{\Lambda}
\DeclareMathOperator{\conv}{conv}
\DeclareMathOperator{\Oh}{\mathcal{O}}

\definecolor{darkblue}{rgb}{0,0,0.7} 

\title{On Triangular Separation of Bichromatic Point Sets}
\author[1]{Helena Bergold}
\author[2]{Arun Kumar Das}
\author[3]{Robert Lauff}
\author[3]{Manfred Scheucher}
\author[4]{Felix Schr{\"o}der}
\author[5]{Marie Diana	Sieper}

\affil[1]{Technische Universit{\"a}t M{\"u}nchen\\helena.bergold@gmail.com}
\affil[2]{Czech Technical University\\arund426@gmail.com}
\affil[3]{Technische Universität Berlin\\lauff@math.tu-berlin.de,scheucher@math.tu-berlin.de}
\affil[4]{Charles University\\schroder@kam.mff.cuni.cz}
\affil[5]{JMU W{\"u}rzburg\\marie.sieper@uni-wuerzburg.de}

\authorrunning{H. Bergold et al.}

\date{}
\begin{document}

\maketitle

\begin{abstract}
We address the problem of computing the minimum number of triangles to separate a set of blue points from a set of red points in $\mathbb{R}^2$. A set of triangles is a \emph{separator} of one color from the other if every point of that color is contained in some triangle and no triangle contains points of both colors. We consider all possible variants of the problem depending on whether the triangles are allowed to overlap or not and whether all points or just the blue points need to be contained in a triangle. We show that computing the minimum cardinality triangular separator of a set of blue points from a set of red points is NP-hard and further investigate worst case bounds on the minimum cardinality of triangular separators for a bichromatic set of $n$ points.
\end{abstract}

\section{Introduction}
\label{sec:intro}
\emph{Separability problems} on colored point sets are concerned with computing a class of objects of given shapes for a given set of colored points, such that every point of a fixed color is covered with the computed objects and no object contains points of different colors. The exact problem varies depending on the underlying optimizing criteria. 
We study triangles separating red and blue points in $\mathbb{R}^2$.
Separating red and blue points from each other is well studied in computational geometry under the name \emph{class cover problems}~\cite{armaselu2017geometric,cannon2004approximation,devinney2003class} due to their wide applications in data mining, pattern recognition, learning theory, and operations research. The problem is formally defined as follows:

\begin{tcolorbox}
	\begin{problem}[\textsc{Overlap-Separation}]
		\label{prb:both_overlap}
  		{\bf Input:} A set of red and blue points $P$ in $\mathbb{R}^2$.
		{\bf Output:}  A minimum set \cT of triangles containing all points such that no member of \cT contains points of both colors.
	\end{problem}   
\end{tcolorbox}

Note that this problem can be subdivided into two problems. First find the minimum number of triangles $\cT_b$ containing all the blue points in $P$ such that no triangle in $\cT_b$ contains a red point. 
Then compute the minimum number of triangles $\cT_r$ containing the red, but no blue points. The set $\cT_b \cup \cT_r$ is an optimal solution for \Cref{prb:both_overlap}. Thus it is enough to solve this:
\begin{tcolorbox}
	\begin{problem}[\textsc{Overlap-Separation-of-blue}]
		\label{prb:blue_overlap}
    	{\bf Input:} A set of red and blue points $P$ in $\mathbb{R}^2$.
		{\bf Output:} A minimum set \cT of triangles containing all blue but no red points of $P$.
	\end{problem}   
\end{tcolorbox}    

However once the triangles are not allowed to intersect each other, the problem becomes more restricted and spawns two different problems for separating the points as follows.   

\begin{tcolorbox}
	\begin{problem}[\textsc{Disjoint-Separation}]
		\label{prb:both_no_overlap}
    	{\bf Input:} A set of red and blue points $P$ in $\mathbb{R}^2$.
		{\bf Output:} A minimum set \cT of disjoint triangles containing all points such that no member of \cT contains points of both colors.
	\end{problem} 
	\begin{problem}[\textsc{Disjoint-Separation-of-blue}]
		\label{prb:blue_no_overlap}
    	{\bf Input:} A set of red and blue points $P$ in $\mathbb{R}^2$.
		{\bf Output:} A minimum set \cT of disjoint triangles containing all blue but no red points of $P$.
	\end{problem}   
\end{tcolorbox}
Researchers addressed the problem of computing the minimum number of separators for a bichromatic point set. 
Canon and Cowen~\cite{cannon2004approximation} first considered the problem of finding the minimum number of circles to separate a bichromatic point set in a general metric space. They proved that computing the minimum number of circles centered at the blue points to separate them from the red ones is NP-hard and presented a $(\ln n+ 1)$-factor approximation algorithm and devised a PTAS for the problem in $\mathbb{R}^d$. Bereg et al.~\cite{bereg2012class} showed that computing the minimum number of axis-parallel rectangles to separate a bichromatic planar point set is NP-hard. In the same paper, they addressed the problem of separating objects by vertical or horizontal strips and presented a $\mathcal{O}(r \log r + b \log b + \sqrt{rb})$-time exact algorithm where $r$ and $b$ are the numbers of red and blue points respectively. They also proved separation to be NP-hard if the separating objects are {half-strips}/squares and presented $\mathcal{O}(1)$-approximations. 
But separation and partition problems are not only interesting on the algorithmic side. Motivated by a question of Aharoni and Saks, Dumitrescu et al.~\cite{dumitrescu2001matching,dumitrescu2001convexpartit,dumitrescu2000matching} showed that every bichromatic set of $n$ points can be partitioned into $\lfloor\frac{n}{2}\rfloor +1$ monochromatic subsets with disjoint convex hulls. This is not true if the subsets have a maximum size of 2. They give an algorithm to find a matching of size $\frac{3}{7}n$, but show that a monochromatic matching can in some cases only cover $\frac{94}{95}n$ points. This is in contrast to the classic problem of Putnam that every bichromatic set of $n$ blue and $n$ red points admits a perfect matching of the red and the blue points (\cite{larson1983problemsolving}, for a proof of some generalization, see for instance \cite{akiyama1989disjsimpl}). For both settings, the special case of a maximum size of a matching stabbed by a line with points on a circle has been of both long-term and recent interest under the name "necklace folding problem"\cite{CSOKA2022123,LyPed99}.

\begin{table}[htb!]
	\begin{tabular}{ r | l | l}
		& cover both & cover blue \\[1ex]  
		\hline & \\[-1.5ex]
		\multirow{2}{6em}{disjoint} & $\geq \lfloor \frac{n}{2} \rfloor +1$ (\Cref{prop:colorchanges}) & $\geq\lfloor\frac{n}{4}\rfloor+1$ (\Cref{prop:no blue triangles})\\ [1ex]
		& $\leq \lfloor \frac{n}{2} \rfloor +1$ (\Cref{prop:induct+comp}) & $\leq\lfloor\frac{2}{7} n\rfloor+1$ (\Cref{prop:blue_disjoint})\\[1ex] 
		\hline & \\[-1.5ex]
		\multirow{2}{6em}{overlap}
		& $\geq\frac{3}{8} n-\Oh(1)$ (\Cref{prop:no blue triangles}) & $\geq\lfloor\frac{n}{4}\rfloor+1$ (\Cref{prop:no blue triangles})\\ [1ex]
		& $\leq \frac{13}{30}n +\Oh(1)$ (\Cref{prop:match+comp})  & $\leq\frac{4}{15}n + \Oh(1)$ (\Cref{prop:combined upper bounds})
	\end{tabular}
	\caption{Summary of bounds on triangles separating bichromatic planar point sets of $n$ points in total.}
	\label{table:results}
\end{table}

We show that \nameref{prb:blue_overlap} and \sloppy \nameref{prb:blue_no_overlap} are NP-hard. Then we prove combinatorial bounds on the number of triangles required for the problems in the worst case for $n$ points, see \Cref{table:results}. The paper is organized in the following manner. Section~\ref{sec:np-hard} describes the hardness of the problems by a polynomial time reduction from the \textsc{Planar Monotone 3-SAT} problem. Section~\ref {sec:bounds} presents an overview of the combinatorial results. 
Finally, the paper is concluded in Section~\ref{sec:concl}. For detailed proofs we refer to the Appendix.



\section{NP-hardness}
\label{sec:np-hard}


We present a polynomial-time reduction from the known NP-hard \textsc{Planar Monotone (PM) 3-SAT}~\cite{de2012optimal} to the problems \nameref{prb:blue_overlap} and \nameref{prb:blue_no_overlap} (Problems \ref{prb:blue_overlap} and \ref{prb:blue_no_overlap}). \textsc{PM 3-SAT} is a special version of \textsc{3-SAT}, where the usual boolean formula $\phi$ is in conjunctive normal form, each clause contains 3 literals and all of them are either positive or negative. Moreover, there exists a planar embedding~$\Gamma_\phi$ of the incidence graph of $\phi$, in which the variables lie on the $x$-axis, the clauses with positive literals are above the $x$-axis and the clauses with negative literals are below the $x$-axis.
Given a PM 3-SAT formula~$\phi$, we construct a set of bichromatic points in the plane such that we can decide whether $\phi$ is satisfiable based on the number of triangles in a minimum separator. The condition holds whether the triangles are allowed to overlap or not. Thus we simultaneously address the NP-hardness of both problems \ref{prb:blue_overlap} and \ref{prb:blue_no_overlap}.

\begin{figure}[b]
 \centering
    \includegraphics[width=.5\textwidth]{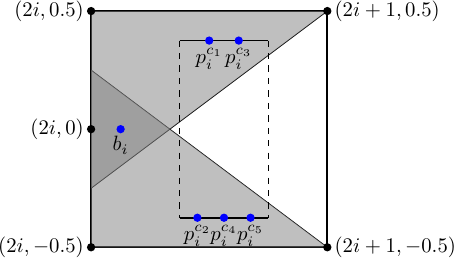}
    \caption{The covering triangles of a variable gadget $x_i$. \robert{The white space between the dashed lines can be used for red points.}}
    \label{fig:var gadget sketch}
\end{figure}

\newcommand{\eps}{\varepsilon}

Here we give an overview of the reduction without presenting the formal details. For the details, see Appendix \ref{sec:ap np-hard}. 
Given a PM 3-SAT formula~$\phi$ with~$k$ variables $x_1,\ldots x_k$ and~$m$ clauses with a corresponding planar embedding~$\Gamma_\phi$ we will construct a bichromatic point set, such that all blue points can be covered with $k+2m$ triangles, if and only if $\phi$ is satisfiable. 
Let $\eps$ be small enough.
First we describe how to replace the variable-vertices of $\Gamma_\phi$ with variable gadgets. The variable $x_i$ is represented by vertices in a square. Inside the square, place two covering triangles intersecting in a small triangle next to the point $(2i,0)$ and place a blue point $b_i$ in their intersection. For every clause $c$ involving $x_i$, place a point $p_{i}^c$ on a top or bottom segment close enough to the middle as in~\Cref{fig:var gadget sketch} making sure that they are ordered according to the order of incidences of the clauses at $x_i$ in $\Gamma_\phi$. Use the top segment if and only if $c$ is positive. Now we replace the vertex associated to the positive clause $c=x_i\vee x_j\vee x_k,i<j<k$ of $\Gamma_\phi$ with 2 blue points $l_c$ and $r_c$, shifted slightly to the left and right, respectively. Place four covering triangles covering each of the pairs of points $(l_c,p_{i}^c),(l_c,p_j^c),(r_c,p_j^c)$ and $(r_c,p_k^c)$ so closely that no other point is inside and triangles only intersect if their pairs do. Then perturb the placed blue points to establish general position. We say that two of the previously placed blue points are \emph{incompatible} if they are not contained in a common covering triangle. For every pair of incompatible points, we place a red point on the segment between them, but not into any covering triangle. If both points are in the same variable gadget, \Cref{fig:var gadget sketch} shows how this is done. If they are not, the segment ``passes through empty space'', that is, the triangles outside of the variable gadgets are thin enough not to contain it completely. In an $\eps$-ball around each blue point, we place $k+2m+1$ extra blue points in such a way that general position of the blue point set is preserved. These points form a so-called \emph{blue point cloud} around the blue point. The $\eps$-ball of every red point $r$ is contained in the convex hull of the $\eps$-balls of its defining blue points, but disconnects it. For every triangle of one point from one of the defining blue point clouds and two of the other, we place a red point inside it and the $\eps$-ball of $r$ and then delete $r$. Thus they are all contained in empty space. We place these points in such a way, that the full point set is in general position. They constitute the \emph{red point cloud} of $r$.

We now show that the thus defined point set can be covered with $k+2m$ triangles if and only if $\phi$ is satisfiable.
Assume that $\phi$ is satisfiable. For every false variable, we use the bottom covering triangle in the variable gadget. For every true variable, we use the top one. For every positive clause $c=x_i\vee x_j\vee x_k$, we use two triangles:
\begin{enumerate}
    \item If $x_i$ is false, we use the triangle covering $(l_c,p_i^c)$.
    \item If $x_k$ is false, we use the triangle covering $(r_c,p_k^c)$. 
    \item If both $x_i$ and $x_k$ are false, $x_j$ is true, so $p_j^c$ has been covered and we are done.
    \item Else if $x_j$ is false we use $(l_c,p_j^c)$ or $(r_c,p_j^c)$, depending on which of the first two triangles was not used, so $l_c$ and $r_c$ are not covered twice.
    \item If we did not cover both $l_c$ and $r_c$ this way, we cover their point clouds individually.
\end{enumerate}
Negative clauses are handled analogously. This covers all blue points with $k+2m$ triangles. Since we only use triangles that are subsets of covering triangles, no triangle contains a red point.

Assume now we found a way to cover the blue points with $k+2m$ triangles. Every point cloud contains more than $k+2m$ points, so some triangle contains at least 2 points of it. This triangle is said to \emph{cover} the cloud. No two incompatible clouds are covered by the same triangle, otherwise it would contain a red point. However the points $l_c,r_c$ for all clauses $c$ and the points $b_i$ for all variables $x_i$ are pairwise incompatible, so we need all $k+2m$ triangles for them, $k$ \emph{variable-covering} and $2m$ \emph{clause-covering}. If the triangle covering $b_i$ covers any $p_i^c$ for a positive clause $c$, set $x_i$ to true, else to false. For any positive clause $c=x_i\vee x_j\vee x_k$, one of the incompatible $p_i^c,p_j^c$ and $p_k^c$ must be covered by the corresponding variable-covering triangle, since no other triangle can cover them and there are only 2 triangles associated with $c$. Hence the clause is fulfilled.
For negative clauses, the argumentation works similarly, since it is impossible that the triangle covering $b_i$ covers the incompatible point clouds $p_i^{c_+},p_i^{c_-}$ for a positive clause $c_+$ and a negative clause $c_-$.

\section{Bounds on the number of triangles}
\label{sec:bounds}


This section shall give you a very rough overview over the techniques used to obtain bounds on the number of triangles needed to separate a bichromatic point set in the worst case. 

For our lower bounds, we use two different point sets: For the case of covering both colors with disjoint triangles, we use points equidistributed along a circle, colored alternatingly with blue and red. This makes sure a triangle can only ever cover three points, however the disjointness of triangles ensures that at least as many triangles as those that reach that bound cover only one point as well. In the disjoint case, this would give a bound of $\frac{n}{3}+\Oh(1)$, but we can do a little bit better by considering a slightly different construction: Here only the blue points are equidistributed along a circle and an almost equal number of red points are placed inside so that any triangle of three blue points contains one of them. This construction gives the $\frac n 4+\Oh(1)$ lower bounds right away, whether the triangles are overlapping or not. Some further investigation reveals that the red points can be chosen such that no five of them can be in a triangle without blue points. This gives the final result of roughly $\frac{1}{2}\frac{n}{2}+\frac{1}{4}\frac{n}{2}=\frac{3}{8}n$ triangles for overlapping triangles covering both point sets.

For the upper bounds, we use a combination of sweeping line algorithms and case distinction at the boundary of the convex hull. In particular, we prove that our lower bound for disjoint covering of all points is tight. We also give bounds for the number of triangles needed to cover all blue points in terms of the number of blue points only, then to the number of red points only in order to arrive at our $\frac{4}{15}n+\Oh(1)$ (overlapping) and $\lfloor\frac{2}{7}n\rfloor+1$ (disjoint) bounds for covering the blue points. Finally, we use computer assistance to deduce that the maximum number of bichromatic points you can place in the plane without two disjoint empty monochromatic triangles is 14. We deduce our $\frac{13}{30}n+\Oh(1)$ bound for covering all points with overlapping disjoint triangles from it.

\section{Conclusion}
\label{sec:concl}
We have addressed the hardness and bounds on the minimum number of triangular separators for a bichromatic point set. Theorem~\ref{thm:NP} proves the NP-hardness of \nameref{prb:blue_no_overlap} but it is questionable whether this problem belongs to NP at all. It might be $\exists \mathbb{R}$-hard~\cite{schaefer2017fixed}.  

\begin{conjecture}
    \nameref{prb:blue_no_overlap} is $\exists \mathbb{R}$-complete.
\end{conjecture}

Table~\ref{table:n<=12} contains our calculations for how many triangles we need to use to cover $n\leq 12$ points in the four different settings.
\begin{table}[htb]
    \centering
    \begin{tabular}{c|c|c|c|c|c|c|c|c|c|c|c|c}
     $n$& 2 &3&4&5&6&7&8&9&10&11&12\\\hline
     one color, overlap & 1&1&2&2&2&2&3&3&3&3&4\\
     one color, disjoint & 1&1&2&2&2&2&3&3&3&3&4\\
     both colors, overlap & 2&2&3&3&4&4&4&5&5&6&6\\
     both colors, disjoint & 2&2&3&3&4&4&5&5&6&6&7\\
    \end{tabular}
    \caption{Results on covering/separating small bichromatic point sets}
    \label{table:n<=12}
\end{table}
Based on these numbers, the OEIS~\cite{OEIS}\todo{reference the corresponding sequences} suggests the following formulas for the overlap setting, one of which coincides with our previously proved lower bound: 

\begin{conjecture}
    Every planar bichromatic set of $n$ points in general position can be covered by at most $\lfloor\frac{2}{5}(n+4)\rfloor$ monochromatic triangles.
\end{conjecture}
In the following conjecture the blue triangles could even possibly be chosen to be disjoint.
\begin{conjecture}
    The blue points in every planar bichromatic set of $n$ points in general position can be covered by at most $\lfloor\frac{n}{4}\rfloor+1$ blue triangles.
\end{conjecture}

The main roadblock to proving this conjecture is that in a set of predominantly blue points, we need some triangles to cover at least 4 blue points. However a triangle that covers 4 blue points in convex position cannot be chosen inside the convex hull of those points. Therefore a similar approach as in the proof of the $\lfloor\frac{n}{2}\rfloor+1$ bound seems out of reach for now.


\bibliography{bibliography.bib}

\newpage 
\appendix

\section{The formal NP-Hardness proof}\label{sec:ap np-hard}

\begin{figure}
\begin{minipage}[t]{.45\textwidth}
    \centering
    \includegraphics[width=\textwidth]{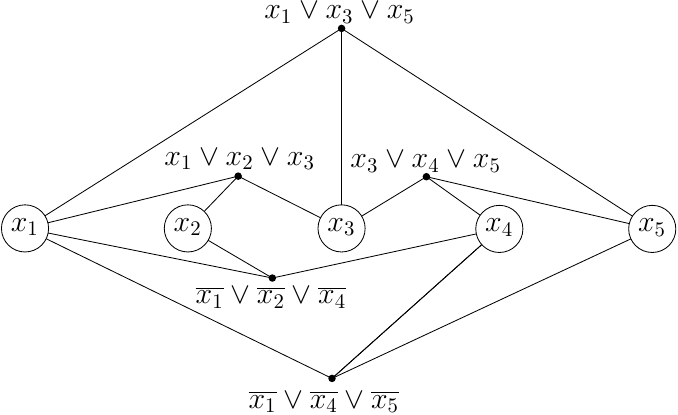}
    \caption{Planar embedding of the PM-3SAT formula $\phi$.
    }
\end{minipage}%
\hspace{.1in}
\begin{minipage}[t]{.45\textwidth}
 \centering
    \includegraphics[width=\textwidth]{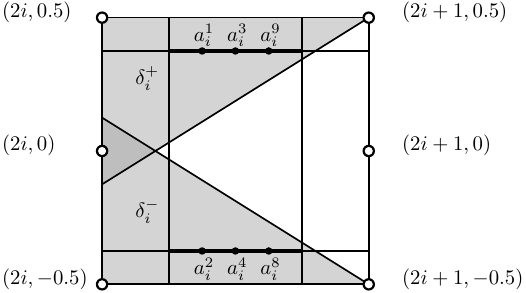}
    \caption{The covering triangles of a variable gadget $x_i$.}
    \label{fig:variable-covering-triang}
\end{minipage}
\end{figure}

\subsection{Construction of the Variable Gadgets.}
Let $x_1, x_2, \dots x_k$ be the variables ordered as they occur left to right in~$\Gamma_\phi$.
For a variable~$x_i$, let~$k_i^+$ and~$k_i^-$ denote the number of positive and negative occurrences of~$x_i$, respectively.
To construct the variable gadget for variable~$x_i$, we construct two covering triangles in the following way, as depicted in \Cref{fig:variable-covering-triang}. We recall that a covering triangle for a set of points refers to the triangle which contains that set of points and does not contain any other point.


The two covering triangles for variable $x_i$ are placed in an axis-parallel square with unit-length, that has its center point at $(2i+0.5,0)$.
We construct the \emph{positive} covering triangle~$\delta_i^+$ which is spanned by the points $(2i, -0.125), (2i, 0.5), (2i+1, 0.5)$; and the \emph{negative} covering triangle $\delta_i^-$ which is spanned by the points $(2i, +0.125), (2i, -0.5), (2i+1, -0.5)$. Note that~$\delta_i^+$ and~$\delta_i^-$ have a non-empty intersection that is to the left of the vertical line defined by~$x=2i+0.25$. We place a point cloud called $V_i$, centered at the blue point $b_i$, referring it as the \emph{defining} point cloud of $x_i$.
We now consider the line segment $x\in[2i+0.25, 2i+0.75],y=0.375$, divide it evenly into $k_i^+ +1$ segments and denote the newly formed $k_i^+$ inner end points of the segments as $x_i$'s positive \emph{anchor} points. Similarly, we divide the line segment $x\in[2i+0.25, 2i+0.75],y=-0.375$ into $k_i^-+1$ into $k_i^- + 1$ segments and denote the newly formed $k_i^-$ inner end points of the segments as $x_i$'s negative anchor points. 
Note that the anchor points will not be blue points, but only help in the construction of the clause gadget.
From left to right, we assign the positive (negative) anchor points of $x_i$ to the positive (negative) clauses in which $x_i$ is contained in the same order in which the corresponding edges from the clauses appear in the cyclic clockwise (counterclockwise) order around $x_i$ in $\Gamma_\phi$. If $x_i$ is contained in a clause $c_j$, we name the anchor point of $x_i$ that is assigned to $c_j$ to be $a_i^j$, referring it as the \emph{variable anchor point} of $c_j$ in $x_i$.
Further we say the right and left neighbour of an anchor point $a_i^j$ are the right and left anchor of $x_i$ on $x\in[2i+0.25, 2i+0.75],y=+0.375$ if $c_j$ is positive or $x\in[2i+0.25, 2i+0.75],y=-0.375$ if $c_j$ is negative. If such a left anchor point does not exist, we say the left neighbour is $(2i+0.25,\pm 0.375)$, and if such a right neighbour does not exist, we say the right neighbour is $(2i+0.75,\pm 0.375)$.

Note that the positive and negative anchor points are contained completely in the inner part of $\delta_i^+$ and $\delta_i^-$ respectively, and all anchor points of $x_i$ are to the right of the intersection of $\delta_i^+$ and $\delta_i^-$.

\subsection{Construction of the Clause Gadgets.}
We only describe the construction of the positive clause gadgets, since the construction of the negative ones is analogue.
During the construction, we will keep the invariant that all covering triangles belonging to a clause gadget of a clause $c$ are contained in the vertical strip bounded by the $x$-coordinates of the leftmost and rightmost variable anchor point of $c$.

Given two positive clauses $c_i, c_j$ with $c_j = (x_t \vee x_u \vee x_v)$ and $t<u<v$, we say $c_i$ is \emph{nested} into $c_j$, if in $\Gamma_\phi$ the vertex corresponding to $c_i$ is drawn inside the area bounded by the $x$-axis and the edges $\{c_j,x_t\}$ and $\{c_j, x_v\}$.
Denote by $m^+$ the number of positive clauses and let $c_1, c_2, \dots, c_{m^+}$ be the positive clauses ordered in such a way that $i<j$ for every clause $c_i$ nested into a clause $c_j$. 
%
%
Let us assume inductively that we already constructed the clause gadgets of $c_1, \dots, c_{j-1}$. Let $C^j_\mathrm{nested}\subseteq \{c_1, \dots, c_{j-1}\}$ be the set of clauses nested into $c_j$. 
Let $y_j$ be one plus the height of a highest point of a covering triangle in the construction of $C^j_\mathrm{nested}$. If $C^j_\mathrm{nested} = \varnothing$, we set $y_j = 1.5$ which corresponds to the height of the variable gadgets plus one.
We will now construct a total of four covering triangles $\triangle_{j,t}^\mathrm{left}, \triangle_{j,u}^\mathrm{left}, 
\triangle_{j,u}^\mathrm{right}, \triangle_{j,v}^\mathrm{right}, 
$ that together form the clause gadget. 
Moreover, for description brevity we denote by $\triangle_{j}^\mathrm{left}$ and $\triangle_{j}^\mathrm{right}$, 
the intersection of $\triangle_{j,t}^\mathrm{left}, \triangle_{j,u}^\mathrm{left}$ and $\triangle_{j,u}^\mathrm{right}, \triangle_{j,v}^\mathrm{right}$, respectively, as shown in \cref{fig:clause-triang-overview}.
In order to construct the triangles, we will first start with constructing cones for each covering triangle into which we want to place the triangles and some minimum heights.

\begin{figure}
    \centering
    \includegraphics{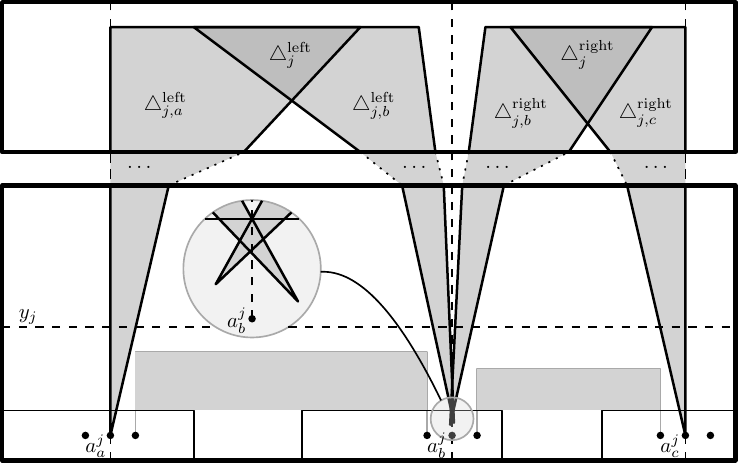}
    \caption{An overview of the covering triangles for clause $c_j$.}
    \label{fig:clause-triang-overview}
\end{figure}

\begin{enumerate}
    \item Construction of the cone for $\triangle_{j,t}^\mathrm{left}$: Let $a_t^j$ be the variable anchor point of clause $c_j$ in $x_a$.
    We define the cone $\Cone_{j,t}^\mathrm{left}$ to have origin~$a_t^j$; as its left boundary the vertical upper half-line starting from $a_t^j$; and as its right boundary the half-line starting from $a_t^j$ and going through the cut of the vertical line through the anchor point that is the right anchor point neighbour of $a_t^j$, and the horizontal line through $y_j$.
    %
    This cone ensures that the covering triangle lying inside it does not intersect any of the covering triangle corresponding to the clauses nested under it.
    Additionally, we also construct two minimum heights, which will later ensure that every line from a point in $\triangle^\mathrm{left}_j$ to a point in $V_t$ or some point cloud to the left of $a_t^j$ will have some segment outside of any covering triangle, which will help in the placement of obstacles later. 
    
    The first minimum height is~$\overline{h}_{j,t}^\mathrm{left}$: 
    Let~$w$ be the point with~$y=1$ above~$a_t^j$. For any point~$p\in V_t$, we construct the line~$\ell_p$ through~$w$ and~$p$. Set $\ell$ to be the $\ell_p$ with the highest slope.
    Then~$\overline{h}_{j,t}^\mathrm{left}$ is the height of the highest intersection point of~$\ell$ with the right boundary of~$\Cone_{j,t}^\mathrm{left}$. Note that such an intersection point may not exist if the slope of~$\ell$ is not lower than the slope of the right boundary of~$\Cone_{j,t}^\mathrm{left}$. In this case, we change the slope of~$\Cone_{j,t}^\mathrm{left}$ until it is a bit higher than the slope of~$\ell$. This makes the cone~$\Cone_{j,t}^\mathrm{left}$ only smaller, so all further statements related to properties of~$\Cone_{j,t}^\mathrm{left}$ remain unaffected.
    The second minimum height is~$\hat{h}_{j,t}^\mathrm{left}$. Let again~$w$ be the point with $y=1$ above~$a_t^j$, and let~$\ell$ be the line through the left neighbor of~$a_t^j$, and~$w$. Then~$\hat{h}_{j,t}^\mathrm{left}$ is the height of the highest intersection point of~$\ell$ with the right boundary of~$\Cone_{j,t}^\mathrm{left}$. Note that such an intersection point always exists, since the slope of this~$\ell$ is always lower than the slope of the right boundary of~$\Cone_{j,t}^\mathrm{left}$.  

    \item Construction of the cone for $\triangle_{j,u}^\mathrm{left}$: Let $a_u^j$ be the variable anchor point of clause $c_j$ in 
    $x_b$, and let $\overline{a}_u^j$ be the point that lies on the vertical line through $a_u^j$, in the middle between $a_u^j$ and where the vertical line through $a_u^j$ cuts the horizontal line $y=1$.  
    We define the cone $\Cone_{j,u}^\mathrm{left}$ to have as its origin~$\overline{a}_u^j$ 
    ; as its right boundary it has the vertical upper half line starting from $\overline{a}_u^j$, and as its left boundary it has the half line starting from $\overline{a}_u^j$ and going through the cut of the vertical line through the anchor point that is the {left anchor neighbour of $a_u^j$}, and the horizontal line through $y_j$.
    \todo{\felix{note to self: resume here}}
    Note that we do not need an analogue to the second minimum height~$\overline{h}_{j,t}^\mathrm{left}$ of (1) here, since this one will only be relevant for certain types of interactions with the variable defining point cloud, which is only on the left side of the anchor points in every gadget.
    However, we again construct a minimum height~$\hat{h}_{j,u}^\mathrm{left}$.
    Let~$w$ be the point with~$y=1$ above~$a_u^j$, and let~$\ell$ be line through the right neighbour of~$a_u^j$, and~$w$. Then~$\hat{h}_{j,u}^\mathrm{left}$ is the height of the highest intersection point of~$\ell$ with the left boundary of~$\Cone_{j,u}^\mathrm{left}$. Note that such an intersection point always exists, since the slope of this~$\ell$ is always lower than the slope of the left boundary of~$\Cone_{j,t}^\mathrm{left}$.

    \item Construction of the cone for $\triangle_{j,u}^\mathrm{right}$:
    Construct $\Cone_{j,u}^\mathrm{right}$ vertically mirrored to the construction of $\Cone_{j,u}^\mathrm{left}$.    
    Construct the minimum height~$\overline{h}_{j,u}^\mathrm{right}$ analogue (not vertically mirrored) to the construction of~$\overline{h}_{j,t}^\mathrm{left}$.
    Construct the minimum height~$\hat{h}_{j,u}^\mathrm{right}$ vertically mirrored to the construction of~$\hat{h}_{j,u}^\mathrm{left}$.

    \item Construction of the cone for $\triangle_{j,v}^\mathrm{right}$:
    Construct $\Cone_{j,v}^\mathrm{right}$ vertically mirrored to the construction of $\Cone_{j,t}^\mathrm{left}$.
    As in (2), we do not need an analogue of minimum height~$\overline{h}_{j,t}^\mathrm{left}$.
    Construct the minimum height $\hat{h}_{j,v}^\mathrm{right}$ vertically mirrored to the construction of $\hat{h}_{j,t}^\mathrm{left}$.
\end{enumerate}

Given the cones $\Cone_{j,t}^\mathrm{left}, \Cone_{j,u}^\mathrm{left}, \Cone_{j,u}^\mathrm{right}$ and $\Cone_{j,v}^\mathrm{right}$, and the corresponding minimum heights, we now construct the covering triangles.
Consider the lowest point $w$ contained both in $\Cone_{j,t}^\mathrm{left}$ and in $\Cone_{j,u}^\mathrm{left}$. If~$w$ lies in the vertical strip bounded by the right neighbour of~$a_t^j$ and the left neighbour of~$a_u^j$, and if it lies above all of the three minimum heights $\overline{h}_{j,t}^\mathrm{left}$, $\hat{h}_{j,t}^\mathrm{left}$ and $\hat{h}_{j,u}^\mathrm{left}$, then we set $i_j^\mathrm{left}$ to be~$w$. 
Otherwise, let $i_j^\mathrm{left}$ be the lowest point in the intersection of $\Cone_{j,t}^\mathrm{left}$ and  $\Cone_{j,u}^\mathrm{left}$ that fulfills these conditions. 
Let $\ell_\mathrm{upper}^\mathrm{left}$ be the horizontal line one unit above $i_j^\mathrm{left}$.
Let the raw triangle $\tilde{\triangle}_{j,t}^\mathrm{left}$ be defined to be the intersection of $\Cone_{j,t}^\mathrm{left}$ with the half plane bounded above by $\ell_\mathrm{upper}^\mathrm{left}$.
We obtain the refined triangle~$\triangle_{j,t}^\mathrm{left}$ out of~$\tilde{\triangle}_{j,t}^\mathrm{left}$ by moving the upper right corner along~$\ell_\mathrm{upper}^\mathrm{left}$ to the left until it is above the left neighbour of~$a_u^j$. (If the upper right corner of~$\tilde{\triangle}_{j,t}^\mathrm{left}$ is already above or to the left of~$a_u^j$, there is nothing to do.)

Let the raw triangle $\tilde{\triangle}_{j,u}^\mathrm{left}$ be defined to be the intersection of~$\Cone_{j,u}^\mathrm{left}$ with $\ell_\mathrm{upper}^\mathrm{left}$. We define $r$ to be the point with $y=0.5$ above $a_u^j$. We define the line $\lambda$ to be the line through $a_u^j$, and the point with $y=0.5$ that has its $x$-coordinate directly between~$a_u^j$ and the right neighbour of~$a_u^j$.
We obtain the refined triangle~$\triangle_{j,u}^\mathrm{left}$ out of~$\tilde{\triangle}_{j,u}^\mathrm{left}$ in the following way: First, we move the upper left corner along $\ell_\mathrm{upper}^\mathrm{left}$ to the right until it is above the right neighbor point of~$a_t^j$ if need be. Then, we rotate the right (vertical) side of ~$\tilde{\triangle}_{j,u}^\mathrm{left}$ in counterclockwise direction around~$r$ while keeping the other two sides on the same lines, until one of the following two events happens:
(i) The upper right corner is on the vertical line through the left neighbour of $a_u^j$; or,
(ii) the lower corner is on $\lambda$.
The result is the desired triangle $\triangle_{j,u}^\mathrm{left}$.

Construct~$\triangle_{j,u}^\mathrm{right}$ and $\triangle_{j,v}^\mathrm{right}$ in an analogue way (vertically mirrored) out of $\Cone_{j,u}^\mathrm{right}$, $ \Cone_{j,v}^\mathrm{right}$ and the three minimum heights for the right side. 

Since both $\triangle_{j,u}^\mathrm{left}$ and $\triangle_{j,u}^\mathrm{right}$ intersect the vertical line $\ell_u^j$ through~$a_u^j$ only inside the positive defining triangle~$\delta_u^+$ of~$x_u$, and since both contain the upper intersection point $r$ of $\ell_u^j$ and the boundary of $\delta_u^+$, the two triangles~$\triangle_{j,u}^\mathrm{left}$ and~$\triangle_{j,u}^\mathrm{right}$ have a non-empty intersection, and this intersection is contained completely in~$\delta_u^+$.

We now place point clouds for $c_j$ in the following way:
In $\triangle_{j}^\mathrm{left}$ and $\triangle_{j}^\mathrm{right}$ we place a point cloud each which we denote $P_{j}^\mathrm{left}$ and $P_{j}^\mathrm{right}$.
We further place a point cloud $O_{j,t}$ in the cut of $\delta_t^+$ and $\triangle_{j,t}^\mathrm{left}$; a point cloud $O_{j,u}$ in the cut of $\delta_u^+$, $\triangle_{j,u}^\mathrm{left}$ and $\triangle_{j,u}^\mathrm{right}$; and a point cloud $O_{j,v}$ in the cut of $\delta_v^+$ and $\triangle_{j,v}^\mathrm{right}$.

\subsection{Placement of Obstacles.}
We now place obstacles to ensure that certain pairs of point clouds cannot be covered by the same triangle in any solution. More precisely, we say two point clouds $P$ and $Q$ are \emph{incompatible} with each other, if in any solution a triangle that contains at least two points out of $P$ can not contain a single point of $Q$ and vice versa.

Let $P$ and $Q$ be point clouds. A \emph{$PQ$-triple} are three points $p, p', p^* \in P\cup Q$ with $\{p, p', p^*\} \cap P \neq \varnothing$ and $\{p, p', p^*\} \cap Q \neq \varnothing$. To make two point clouds $P, Q$ incompatible with each other, we place an obstacle in the triangle spanned by $p, p', p^*$ for every $PQ$-triples~$p, p', p^*$. Since no obstacle is allowed to lie inside some covering triangle, we need to show that the interior of the triangle spanned by $p, p', q$ contains at least one point (in the following called \emph{outside point}) that is not contained in any covering triangle. Since the covering triangles are closed areas, the area without the covering triangles is open, thus we can still guarantee that the obstacles can be placed in general position.

We make all pairs of point clouds $P, Q$ incompatible, for which there exists no covering triangle that contains both $P$ and $Q$. Lemma \ref{lem:insert-obstacles} assures that this can always be done.

\begin{restatable}{lemma}{lemInsertObstacles}
    \label{lem:insert-obstacles}
    Given the above construction and let $m_b$ be the number of blue points it contains. It is possible to insert a set of $\mathcal{O}(m_b^3)$ obstacles in such a way, that all pairs of point clouds~$P, Q$, for which there exists no covering triangle that contains both~$P$ and~$Q$, are incompatible.
\end{restatable}

\begin{proof}
We make all pairs of point clouds $P, Q$ incompatible, for which there exists no covering triangle that contains both~$P$ and~$Q$. The following cases can appear:
\begin{enumerate}
    \item $P$ and~$Q$ are occurrence point clouds of two different clauses in the same variable gadget, such that one of the clauses is positive and one of the clauses is negative.
    
    Then, for every $PQ$-triple~$p, p', p^*$ the triangle spanned by the triple contains in it's interior an open subinterval of the $x$-axis that is not contained in any covering triangle, which gives us an outside point.
    
    \item $P$ and~$Q$ correspond to point clouds in gadgets of different variables~$x_i, x_j$, where~$P$~($Q$) can (independent of each other) be the defining point cloud $V_i$ ($V_j$) or an occurrence point cloud of $x_i$ ($x_j$).

    Then, for every $PQ$-triple~$p, p', p^*$ the triangle spanned by the triple contains in it's interior a point in the space between two variable gadgets, which gives us an outside point.
    
    \item $Q$ is either the defining or some occurrence point cloud of a variable~$x_i$; and~$P$ is a defining point cloud of a clause gadget~$c_j$ that does not contain the variable~$x_i$.  
    
    Then, for every $PQ$-triple~$p, p', p^*$ the triangle spanned by the triple contains in it's interior an (open) line segment~$\ell$ from the convex hull of~$P$ to the inside of the convex hull of~$Q$, which cuts a defining triangle of~$Q$ and a defining triangle of~$P$, which are distinct by construction. 
    If the point where~$\ell$ leaves the last defining covering triangle of~$c_j$ is not contained in any other covering triangle, we have found an outside point.
    If not, the covering triangle must belong to a variable~$x_l$ with~$l \neq i$. In this case, the rest of~$\ell$ is contained only in the horizontal strip in which the variables gadgets are contained, and it must leave the unit square of variable~$x_l$ somewhere, in the vicinity of which we find an outside point.
    
    \item $Q$ is the defining point cloud $V_i$ of a variable $x_i$; and $P$ is a defining point cloud of a clause gadget $c_j$ that contains $x_i$.
    
    Assume without loss of generality that $c_j$ is a positive clause with $c_j = (x_a \vee x_b \vee x_c), a<b<c$.
    For every $PQ$-triple~$p, p', p^*$ the triangle spanned by the triple contains in it's interior an (open) line segment $\ell$ from the convex hull of $P$ to the convex hull of $Q$. Let $w$ be the point on the the boundary of the square containing the variable gadget of $x_i$ that is cut by $\ell$.
    If there exists an outside point directly before $w$ on $\ell$ we are done.
    
    Thus, assume that this is not the case. Then, $w$ is contained in the segment $x\in (2i+0.25, 20+0.75), y=0.5$, and we have one of the three cases: Case (i): $i = a$ and $P = P_j^\mathrm{left}$; Case (ii): $i = a$ and $P = P_j^\mathrm{right}$; or Case (iii): $i = b$ and $P = P_j^\mathrm{right}$
    This is exhaustive, since if it were none of these cases, then $\ell$ would not go to the left if starting from $P$.
    
    Case (i): $i = a, P = P_j^\mathrm{left}$. Every point in $P$ and thus also the starting point of $\ell$ lies above the minimum height~$\overline{h}_{j,a}^\mathrm{left}$ and inside the cone $\Cone_{j,a}^\mathrm{left}$. This together with the fact that $\ell$ ends on the convex hull of $Q$ gives us that $\ell$ cuts the vertical line through the anchor point $a_a^j$ truly above $y=0.5$ and we find an outside point on $\ell$ directly to the left of this cutting point.

    Case (ii): $i = a, P = P_j^\mathrm{right}$. All covering triangles that contain~$P_j^\mathrm{right}$ are area-wise distinct from the unit square that contains the variable gadget of~$x_a$, and we can proceed analogue to~(3).
    
    Case (iii): $i = b, P = P_j^\mathrm{right}$. Every point in $P$ and thus also the starting point of $\ell$ lies above the minimum height~$\overline{h}_{j,b}^\mathrm{right}$ and inside the cone~$\Cone_{j,b}^\mathrm{right}$. This together with the fact that~$\ell$ ends on the convex hull of~$Q$ gives us that~$\ell$ cuts the vertical line through the anchor point $a_b^j$ truly above $y=0.5$, and is thus an outside point.
    
    \item $Q$ is an occurrence point cloud of a variable $x_i$ assigned to a clause $c_l$; and $P$ is a defining point cloud of a clause gadget $c_j$ with $j \neq l$ that does contains $x_i$, such that $c_l$ and $c_j$ are of a different type (one positive and one negative).
    
    Assume without loss of generality that $c_j$ is positive.
    Then, for every $PQ$-triple~$p, p', p^*$ the triangle spanned by the triple contains in it's interior an (open) line segment $\ell$ from the convex hull of $P$ to the convex hull of $Q$. Let $w^\ell_{0.5}$ be the cutting point of $\ell$ with the line $y=0.5$, and let $w^\ell_0$ bet the cutting point of $\ell$ with the $x$-Axis. 
    
    If there is an outside point on $\ell$ shortly above $w^\ell_{0.5}$, we are done.
    
    If the point on $\ell$ shortly above $w^\ell_{0.5}$ is not an outside point, then $w^\ell_{0.5}$ is on the upper border of the unit square containing the variable gadget of some variable $x_r$; more precisely, $w^\ell_{0.5}$ is on the line segment defined by $x\in[2r+0.25,2r+0.75], y=0.5$, which is the line segment above the anchor points of $x_r$.
    Now, if $i=r$, then $w^\ell_{0}$ lies on the line segment defined by $x\in[2i+0.25,2i+0.75], y=0$, and thus $w^\ell_{0}$ is an outside point.
    
    Otherwise $i\neq r$, and $\ell$ leaves the unit square containing the variable gadget of $x_r$ inside the horizontal strip defined by $-0.5\leq y \leq 0.5$. Thus we find an outside point directly after the last intersection of $\ell$ with the boundary of the square containing the variable gadget of $x_r$.
    
    \item $Q$ is an occurrence point cloud~$O_i^l$ of a variable $x_i$ assigned to a clause~$c_l$; and~$P$ is a defining point cloud of a clause gadget $c_j$ with~$j \neq l$ that does contains~$x_i$, such that~$c_l$ and~$c_j$ are of the same type (both positive or both negative).
    
    Assume without loss of generality that~$c_j$ is a positive clause with~$c_j = (x_a \vee x_b \vee x_c), a<b<c$. For every $PQ$-triple~$p, p', p^*$ the triangle spanned by the triple contains in it's interior an (open) line segment $\ell$ from the convex hull of~$P$ to the convex hull of~$Q$.
    
    If~$P$ is not contained in some covering triangle together with the occurrence point cloud~$O_i^j$, we can proceed analogue to~(3).
    
    Otherwise, if the clause~$c_l$ is nested in the clause~$c_j$, we find an outside point directly after the last common point of~$\ell$ with the union of the three covering triangles containing~$P$.
    
    Finally, assume that~$c_l$ is not nested in~$c_j$. Then,~$i\neq b$, or in other words~$x_i$ is one of the two out variables of~$c_j$. Without loss of generality, assume that~$i=a$ and~$P=P_j^\mathrm{left}$.
    Since every point in~$P$ and thus also the starting point of~$\ell$ lies above the minimum height~$\hat{h}_{j,a}^\mathrm{left}$ and inside the cone~$\Cone_{j,a}^\mathrm{left}$, and since~$\ell$ ends on the convex hull of~$Q$, we know that~$\ell$ cuts the vertical line through the anchor point~$a_a^j$ truly above~$y=0.5$ and we find an outside point on~$\ell$ directly to the left of this cutting point.

    \item $Q$ is the occurrence point cloud of a clause~$c_j$ in a variable~$x_i$, and~$P$ is a defining point cloud of $c_j$ such that there exists no covering triangle that covers both~$P$ and~$Q$. More precisely, if $c_j = (x_a \vee x_b \vee x_c)$ (or $c_j = (\overline{x}_a \vee \overline{x}_b \vee \overline{x}_c)$) with $a<b<c$, then either $i=a, Q=O_{j,a}, P=P_j^\mathrm{right}$ or $i=c, Q=O_{j,c}, P=P_j^\mathrm{left}$.
    
    If $a=i, Q=O_{j,a}, P=P_j^\mathrm{right}$, then~$Q(=O_{j,a})$ is to the left of the anchor point~$a_b^j$; and if $c=i, Q=O_{j,c}, P=P_j^\mathrm{left}$, then $Q(=O_{j,c})$ is to the right of the anchor point~$a_b^j$.
    Thus, we can proceed as in~(6).
    
    \item $P$ and~$Q$ are the two defining point clouds of a clause. In other words there exists a clause $c_j = (x_a, x_b, x_c)$ (or $c_j = (\overline{x_a}, \overline{x}_b,\overline{x}_c)$) with~$a < b < c$, such that~$P = P_{j}^\mathrm{left}$ and~$Q = P_{j}^\mathrm{right}$.
    
    Then, for every $PQ$-triple~$p, p', p^*$ the triangle spanned by the triple contains in it's interior a point on the vertical line through the anchor point~$a_b^j$, which is an outside point.
    
    \item $P$ is a defining point cloud of a clause~$c_j$ and~$Q$ is a defining point cloud of a clause $c_{j^*}$ with $j \neq j^*$, such that~$c_j$ and~$c_{j^*}$ are of the same type (both positive or negative).
    
    For every $PQ$-triple~$p, p', p^*$ the triangle spanned by the triple lies either completely above the unit squares containing the variable gadgets or completely below them. Since the triangles spanned by ~$p, p', p^*$ contains both some area of a covering triangle of~$c_j$ and a covering triangle of~$c_{j^*}$, and since due to construction, the covering triangles of different clauses are non-overlapping, there is an outside point in between.
    
    \item $P$ is a defining point cloud of a clause~$c_j$ and~$Q$ is a defining point cloud of a clause~$c_{j^*}$ with $j\neq j^*$, such that~$c_j$ and~$c_{j^*}$ have a different type (one is positive and one is negative).
    
    Then, for every $PQ$-triple~$p, p', p^*$ the triangle spanned by the triple contains in it's interior an open non-vertical line segment $\ell$ from the convex hull of~$Q$ to the convex hull of~$P$.

    If the intersection of $\ell$ with the $x$-axis is an outside point, we are done.
    Otherwise, the line intersects the $x$-axis in the cut of the two defining triangles of some variable~$x_i$. Since $\ell$ is non-vertical, we can follow it from it's intersection with the $x$-axis to the left until we leave the unit square that contains the variable gadget of $x_i$, where we find an outside point.
\end{enumerate} 

In this manner, we have inserted at most one obstacle for every triple of blue points, thus the number of obstacles is bounded in $\mathcal{O}(m_b^3)$.
\end{proof}

We can now show the equivalence between the PM-3-SAT formula $\phi$ and the constructed point set $\mathcal{P}$ seen as an instance of Problem \ref{prb:blue_overlap} or Problem \ref{prb:blue_no_overlap}.

\begin{restatable}{lemma}{lemequivalencepointsphi}
    \label{lem:NP}
    Let $\phi$ be a PM 3-SAT formula with $k$ variables and $m$ clauses. Then $\phi$ is satisfiable if and only if the blue points in $\mathcal{P}$, constructed from $\phi$ as described above, are separable from the red points by $k+2m$ triangles. This is true even if the triangles need to be disjoint.
\end{restatable}

\begin{proof}
    Assume that $\phi$ is satisfiable, and let $X \in \{\mathrm{True}, \mathrm{False}\}^k$ be a solution of $\phi$.
    For every variable $x_i$, we take the triangle $\delta_i^+$ to the solution if $X[i] = \mathrm{True}$ and the triangle $\delta_i^-$ to the solution if $X[i] = \mathrm{False}$, with which we cover the point cloud $V_i$, and all point clouds corresponding to a clause containing $x_i$ that are fulfilled by the truth-assignment of $x_i$.
    Now let $c_j = (x_a \vee x_b \vee x_c)$ with $a<b<c$ be a positive clause.
    Since at least one of the point clouds $O_{j,a}$, $O_{j,b}$, $O_{j,c}$ corresponding to the occurrence of $x_a, x_b, x_c$ in $c_j$, respectively, must already be covered by the solution, only the following cases can occur:
    \begin{enumerate}
        \item $O_{j,a}$ is already covered, $O_{j,b}$, $O_{j,c}$ are not: Take $\triangle_{j,b}^\mathrm{left}$ and $\triangle_{j,c}^\mathrm{right}$ into the solution.
        \item $O_{j,b}$ is already covered, $O_{j,a}$, $O_{j,c}$ are not: Take $\triangle_{j,a}^\mathrm{left}$ and $\triangle_{j,c}^\mathrm{right}$ into the solution.
        \item $O_{j,c}$ is already covered, $O_{j,a}$, $O_{j,b}$ are not: Take $\triangle_{j,a}^\mathrm{left}$ and $\triangle_{j,b}^\mathrm{right}$ into the solution.
        \item $O_{j,a}$, $O_{j,b}$ are already covered, $O_{j,c}$ is not: Take $\triangle_{j}^\mathrm{left}$ and $\triangle_{j,c}^\mathrm{right}$ into the solution.
        \item $O_{j,a}$, $O_{j,c}$ are already covered, $O_{j,b}$ is not: Take $\triangle_{j,b}^\mathrm{left}$ and $\triangle_{j}^\mathrm{right}$ into the solution.
        \item $O_{j,b}$, $O_{j,c}$ are already covered, $O_{j,a}$ is not: Take $\triangle_{j,a}^\mathrm{left}$ and $\triangle_{j}^\mathrm{right}$ into the solution.
        \item $O_{j,a}$, $O_{j,b}$, $O_{j,c}$ are all already covered: Take $\triangle_{j}^\mathrm{left}$ and $\triangle_{j}^\mathrm{right}$ into the solution.
    \end{enumerate}
    In all of the cases, we have added two triangles to the solution with which all leftover occurrence point clouds of $O_{j,a}$, $O_{j,b}$, $O_{j,c}$ and additionally the two clause defining point clouds $P_{j}^\mathrm{left}$, $P_{j}^\mathrm{right}$ are covered.
    We proceed analogue for the negative clauses.
    Now, the solution contains $k+2m$ non-overlapping triangles, and all point clouds and thus all blue points are covered by them. Further, since all obstacles were places outside of the covering triangles, the solution separates the blue points from them. Thus, we have found a valid solution of the point set $\mathcal{P}$ constructed above of the claimed size.

    Conversely, assume that there exists a solution $\mathcal{T}$ for the constructed point set $\mathcal{P}$ consisting of $k+2m$ (possibly overlapping) triangles.
    Let $\mathcal{Q} = \{V_i \mid i = 1, \dots, k\} \cup \{P_{j}^\mathrm{left}, P_{j}^\mathrm{right} \mid j = 1, \dots, n\}$. Then $|\mathcal{Q}| = k+2m$ and due to construction, all point clouds in $\mathcal{Q}$ are pairwise incompatible with each other.
    Let $Q \in \mathcal{Q}$. Since $|Q| = k+2m+1 > k+2m = |\mathcal{T}|$, there must be a triangle $\triangle_Q \in \mathcal{T}$ that covers at least two points in $Q$ and thus no other point in $\bigcup_{Q' \in \mathcal{T}, Q'\neq Q} {Q'}$. By repeating this argument for all $Q' \in \mathcal{T}, Q'\neq Q$, we can see that no other triangle in $\mathcal{T}$ apart from $\triangle_Q$ covers any points in $Q$, thus $Q$ is covered completely by $\triangle_Q$.
    This implies a unique one-to-one correspondence between $\mathcal{Q}$ and $\mathcal{T}$.
    We construct a solution vector $X \in \{\mathrm{True}, \mathrm{False}\}^k$ of $\phi$in the following way:
    For every variable $x_i$ let $\triangle_{V_i}\in \mathcal{Q}$ be the unique element in $\mathcal{Q}$ that covers $V_i$. If there exists a positive clause $c_j$ containing the variable $x_i$, such that $\triangle_{V_i}$ contains at least two points of the occurrence point cloud $O_{j,i}$, we set $X[i] = \mathrm{True}$. Note that in this case due to incompatibility, $Q$ contains no point in a point cloud corresponding to a negative occurrence of $c_j$.
    If on the other hand the above condition is not fulfilled, we set $X[i] = \mathrm{False}$.
    Assume there is an unfulfilled clause $c_j = (x_a \vee x_b \vee x_c)$ (or $c_j = (\overline{x_a} \vee \overline{x_b} \vee \overline{x_c})$).
    Since $c_j$ is unfulfilled, none of the triangles $\triangle_{V_a}, \triangle_{V_b}, \triangle_{V_c}$ covers two or more points of the occurrence point clouds $O_{j,a}, O_{j,b}$ or $O_{j,c}$, respectively. 
    But due to pairwise incompatibility of $O_{j,a}, O_{j,b}$ and $O_{j,c}$, there exist three pairwise distinct triangles $t_a, t_b$ and $t_c\in \mathcal{Q}$ that each cover at least two points of $O_{j,a}, O_{j,b}$ and $O_{j,c}$, respectively.
    But the only two point clouds in $\mathcal{Q}$ that are compatible with $O_{j,a}, O_{j,b}$ or $O_{j,c}$ are $\triangle_{P_{j}^\mathrm{left}}$ and $\triangle_{P_{j}^\mathrm{right}}$, a contradiction.
    Thus, $X$ is a valid truth-assignment for $\phi$, which concludes the proof.
\end{proof}

The only remaining piece to establish the NP-hardness of \nameref{prb:blue_overlap} and \nameref{prb:blue_no_overlap} is to show that the construction of $\mathcal{P}$ can be done in polynomial time.
\begin{restatable}{lemma}{lemPolySize}
    \label{lem:poly-size}
    Let $\phi$ be a PM 3-SAT formula with $k$ variables and $m$ clauses. The point set $\mathcal{P}$ together with their coordinates, constructed from $\phi$ as described above, has polynomial size in $k$ and $m$.
\end{restatable}
\begin{proof}
    We first show the total number of points we place is polynomial in terms of $k$ and $m$.
    Each point cloud contains $2k+m+1$ blue points, and we placed $(k + 5m)$ point clouds: a point cloud $V_i$ for every variable $x_i$, and the five point clouds $P_{j}^\mathrm{left}$, $P_{j}^\mathrm{right}, O_{j,a}$, $O_{j,b}$, $O_{j,c}$ for every clause $c_j = (x_a \vee x_b \vee x_c)$ or $c_j = (\overline{x_a} \vee \overline{x_b} \vee \overline{x_c})$.
    The total number of blue points is thus $m_b = (2k+m+1)\cdot (k + 5m)$.
    The number $m_r$ of obstacles is upper bounded in $\mathcal{O}(m_b^3)$ since for every triple of blue points we placed at most one obstacle.

    In order to show that the coordinates of the points are polynomial in terms of $k$ and $m$ as well, it suffices to show that the coordinates of the endpoints of covering triangles are polynomial in $k$ and $m$. This clearly holds for the two covering triangles per variable gadget.
    The end-points of the triangle corresponding to clause gadgets are all constructed out of them by $\mathcal{O}(k+m)$ steps, and in each of the steps the size of the coordinate in can only increase by a constant factor. Thus, the length of the binary representation can only increase by a constant additive term in each step, and the whole instance is of polynomial size.
    Therefore, the construction of the point set can be done in polynomial time.
\end{proof}

This implies the following theorem:

\begin{theorem}
    \label{thm:NP}
     \nameref{prb:blue_overlap} and \nameref{prb:blue_no_overlap} are NP-hard.
\end{theorem}
\begin{proof}
    Follows directly from \ref{lem:NP} and \ref{lem:poly-size}.
\end{proof}

We note that the precise complexity in which \nameref{prb:blue_no_overlap} is contained is unclear; however, for \nameref{prb:blue_overlap}, this result can be strengthened.

\begin{theorem}
    \label{thm:NP-complete}
    \nameref{prb:blue_overlap} is NP-complete.
\end{theorem}
\begin{proof}
    We know that \nameref{prb:blue_overlap} is NP-hard from \Cref{thm:NP}.

    We show that \nameref{prb:blue_overlap} is also contained in NP.
    Let $\mathcal{P}$ be a set of blue and red points.
    
    First, we argue that it suffices to consider only (possibly degenerate) triangles that have at least two points (we allow both red and blue) on each side segment.
    If we have a pseudo-solution consisting only of triangles of this type, the only thing that hinders it from being a real solution for the problem are that there are red points on the borders of the triangles. However, since the points lie in general position, we can always shrink this triangle such that it contains precisely the same blue points as before, but none of the red border points.
    The other way around, if we have any solution $\mathcal{T}$, we can transform every triangle, by taking each of the three sides and moving and/or rotating it individually until it hits 2 points while never losing any blue point in the set. In this manner, red points can appear only at the border of the triangle.

    With this categorization of a solution (or at least some type of pseudo-solution out of which we are guaranteed to be able to generate a solution), we construct all triangles that can be a part of such a solution in the following way:

    We compute the set $\mathcal{L}$ of lines that contain  precisely two points out of $\mathcal{P}$. Clearly, $|\mathcal{L}| \in O(|\mathcal{P}|^2)$
    We can then compute the set of candidate triangles (including the degenerate triangles) by taking all subsets of $\mathcal{L}$ of size 3 or 1, and set them as the bounding lines of the respective triangle.
    This results in $O(|\mathcal{L}|^3)$ = $O(|\mathcal{P}|^6)$ triangles, out of which we can guess the right solution in nondeterministic polynomial time.    
\end{proof}



\section{Detail results of \Cref{sec:bounds}}

\begin{restatable}{proposition}{propcolorchange}
	\label{prop:colorchanges}
	Let $k \geq 0$. Given a bichromatic point set on a closed convex curve in the plane such that $2k$ is the number of times two consecutive points on the curve have different colors, there is no set of $k$ disjoint triangles containing all points such that no triangle contains points of both colors.
\end{restatable}


\begin{proof}
	We will prove the proposition by induction. Note that it holds for $k=0$, because that means that only one color of points occurs and we must use a single, big enough triangle.
	Assume now $k>0$. 
	We start with a triangle to cover some of the points. Then this triangle cuts the curve into at most three pieces. 
	We assume that there are exactly three pieces, where some pieces of the curve might not contain additional points. 
	We can make those three pieces into closed convex curves again by routing a connection along the boundary of the cutting triangle. 
	Let $2k_1,2k_2,2k_3$ be the number of color changes on these three parts. If there are no points, we set $k_i = -1$ since we clearly do not need additional triangles to cover those parts. 
	By the convexity of the curve every triangle that covers some of the points can only cover points from one of the pieces, because otherwise it would intersect the cutting triangle. 
	Therefore by induction we need at least $k_1+1+k_2+1+k_3+1$ triangles to cover the whole point set. 
	We now know that any of the color changes either appeared on one of the three pieces or exactly one of the points was covered by the cutting triangle.
	In this case, the cutting triangle cuts the curve in between the two consecutive vertices of different colors. Since a triangle can cut a convex curve at most six times, we get:
	\begin{equation*}
	   2k\leq 2k_1+2k_2+2k_3+6\iff 1+k_1+1+k_2+1+k_3+1\geq k+1.
    \end{equation*}
	However the latter is the total number of triangles we need, one for the original cutting triangle plus the number of triangles we still need to cover the three pieces, and it is at least $k+1$. This concludes the proof.
\end{proof}

\begin{restatable}{proposition}{propinductcomp}
   Given a bichromatic planar $n$-point set, there are $\lfloor\frac{n}{2}\rfloor + 1$ disjoint triangles within the convex hull of $P$ containing all points such that no triangle contains points of both colors.
	\label{prop:induct+comp}
\end{restatable}


\begin{proof}
	In a first step, we show that it is sufficient to prove the statement for point sets with an odd number of points. 
	Let $P$ be a set with an even number $n$ of points. Let $x$ be a point from the convex hull of $P$ and $P' = P\backslash \{x\}$. Then $P'$ contains an odd number of points.
	Using the odd case, we can cover $P'$  with $\frac{n}{2}$ triangles which stay in the convex hull of $P'$. Since $x$ is not in the convex hull of $P'$, we can cover $x$ itself with an additional triangle. Hence in total we can cover $P$ with $\frac{n}{2} +1 $ triangles. 
	
	To show the odd case, we choose two points $a,b$ with the same color at least one of which is from the convex hull of $P$. 
	If we divide the point set $P$ by the line  $ab$, we obtain two smaller point sets. One of the point sets $P_o(a,b)$ has an odd number of points. The other one has an even number. We count $a$ and $b$ to the part $P_e$ with an even number of points.
	$P_e(a,b)$ has the following property:
	\begin{itemize}
		\item There are two points of the same color on the convex hull. 
	\end{itemize}
	
	As we show in the following, except for a few \emph{bad examples}, those point sets can be covered with $\frac{n}{2}$ triangles. 
	The bad examples can be described as follows (in particular $n=4,6$): 
	\begin{enumerate}
		\item Three points of the one color whose convex hull contains a point of the other color. \label{item:badexample4}
		\item Three red and three blue points such that the convex hulls of the three points of each color contain a point of the other color. \label{item:badexample6}
	\end{enumerate}

    Moreover, for $|P|=n \leq 12$, we checked these claims by computer using a SAT solver, see \Cref{sec:computational}. 

    We prove the statements in the more general setting of pseudoconfiguration of points which are described via triple orientations of points. Moreover in order to study it with computer assistance we restrict to the cases where the triangles are spanned by at most 3 points of the point set. 
    A  detailed description of the SAT model is given in \cref{sec:computational}. 
	\begin{enumerate}
		\item If $n$ is odd, $P$ can be covered with $\frac{n+1}{2}$ triangles contained in $\conv(P)$.
		\item If $n$ is even and there are two consecutive points of the same color on the convex hull. Then $P$ can be covered with $\frac{n}{2}$ triangles contained in $\conv(P)$ if $P$ is not a bad example.
	\end{enumerate}
	
	For $n >12$ we show the two statements by induction. 
	First, let $P$ be a point set of $n=2k$ bichromatic points with two consecutive points $a$ and $b$ on the convex hull with the same color, without loss of generality blue. 
	Let $P' = P \setminus \{a,b\}$.
	We say a point $x\in P'$ is \emph{visible} from $a$ and $b$ if the triangle $abx$ does not intersect the convex hull of the remaining points. 
	
	If $P'$ contains two consecutive points of the same color on the convex hull, we cover $a$ and $b$ by a triangle. Since $n >12$, $P'$ has at least $10$ points and hence is not one of the bad examples.
	By induction, we can cover $P'$ by $\frac{n}{2}-1$ triangles. Hence $P$ can be covered with $\frac{n}{2}$ triangles. 
	
	Moreover, if there is one point $x$ on the convex hull of $P'$ which is blue and visible from $a$ and $b$, the triangle $abx$ does not intersect $\conv(P'\setminus\{x\})$.
	The number of remaining points is $2k-3$ which is odd and smaller than $2k$. By induction we can cover those points by $\frac{2k-2}{2}$ triangles.
	Hence in total $\frac{2k}{2} = \frac{n}{2}$ triangles cover $P$ as claimed. 
	
	\begin{figure}
		\centering
		\includegraphics[scale = 0.6]{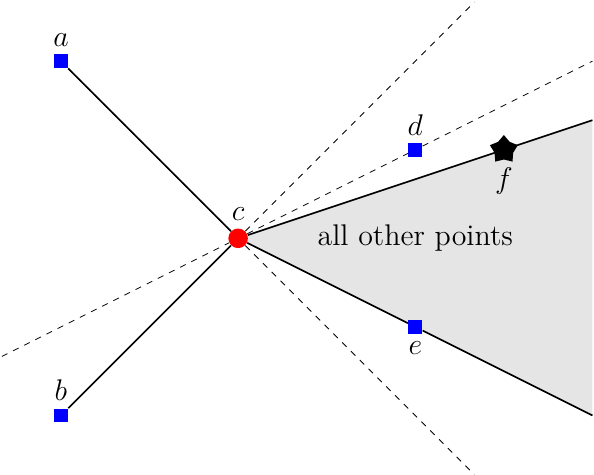}
		\caption{A visualization for the proof of \Cref{prop:induct+comp}.}
		\label{fig:All points hidden}
	\end{figure}
	
	Hence all points $x$ on the convex hull of $P'$ that are visible from $a$ and $b$ must be red. Since there is always at least one such point and those points form a consecutive part of the convex hull, which has no two consecutive points of the same color, 
	there is exactly one visible point $c$ from $a$ and $b$ which is red. 
	For illustration refer to \Cref{fig:All points hidden}. 
	
	Let $d$ and $e$ be the points next to $c$ on the convex hull of $P'$. 
	Since no two consecutive points on the convex hull of $P'$ have the same color, $d$ and $e$ are both blue. 
	Let $\ell$ be the line spanned by $c$ and $d$. Let $f$ be the first of the not yet labeled points that is hit if we rotate it around $c$ towards the remaining points of $P'$. 
    Since all points of $P'$ are on the same side of the line $\ell$ and we rotated towards $P'$ all points of $P' \setminus \{d\}$ are on the same side of the line spanned by $c$ and $f$. 
	Since $a$ and $b$ are separated by the rotated line one of the two points is on the same side as $P' \setminus \{d\}$. We assume without loss of generality $b$ is on the same side. 
	
	Hence $c$ and $f$ are two consecutive points on the convex hull of $P \setminus \{a,d\}$. 
	If $f$ is red, we can cover $a$ and $d$ by a triangle and proceed by induction on the point set $P \setminus \{a,d\}$ which has two consecutive points on the convex hull of the same color. Hence by induction we can cover $P$ with $\frac{n}{2}$ triangles. 
	
	Otherwise $f$ is blue and the three blue points $a,d,f$ span an empty triangle. 
	This triangle does not intersect the convex hull of the points in $P \setminus \{a,d,f\}$. By induction for the odd case, we can cover those points with $\frac{n-2}{2}$ triangles, which yields $\frac{n}{2}$ triangles to cover $P$. 
	This completes the proof for the even case. \\

	Now let $P$ be a point set with an odd number $n$ of bichromatic points.
	Without loss of generality there are more blue than red points. 
    If there are two consecutive points of the same color on the convex hull, we can cover those with a triangle. The remaining point set can be covered with $\frac{n-1}{2}$ triangles that do not overlap with this triangle by induction, thus we need at most $\frac{n+1}{2}$ triangles to cover $P$. 
    
    Thus there is a blue point $a$ on the convex hull.
	Every pair $(a,b)$ of blue points defines a line that splits the point set into two sets, one with an even number of points and one with an odd number of points. Note that every bad example has exactly 1 or exactly 3 blue points. Since the even side includes $a$ and $b$ it has exactly one further blue point if it is a bad example.

	Assume for contradiction that every pair $(a,b)$ of blue points defines a bad even side. Fix some blue point $b\neq a$. If the odd side of $(a,b)$ contains no blue points then we have 3 blue points and hence a point set with at most 5 points in total and the claim follows by the induction base. 
	Thus the odd side contains a blue point. Rotate the line through $a$ and $b$ around $a$ towards the odd side until the next blue point $b'$ is hit. The even side defined by $(a,b')$ is again bad. The even side is not the one containing $b$ though, since it would contain the even side of $(a,b)$ and therefore another blue point. But then all blue points are contained either in the bad example defined by $a,b$ or $a,b'$. As these both only contain at most 3 blue points each, and they even share the point $a$, there are at most 5 blue points. As there are more blue than red points, we get $n\le 9$, all of which are cases we analysed by computer.\\
	It follows that we can find $(a,b)$ such that the even side is not a bad example. Covering the resulting sides using induction yields a covering with the claimed number of triangles.
\end{proof}

\begin{restatable}{proposition}{propmatchcomp}
   \label{prop:match+comp}
  Given a bichromatic planar $n$-point set, there are $\frac{13}{30}n+\Oh(1)$ triangles containing all points such that no triangle contains points of both colors.
\end{restatable}


\begin{proof}
    As we discuss in \Cref{sec:computational}, any point set of 15 blue and/or red points contains two vertex-disjoint monochromatic triangles. Thus we find at least $2\left\lfloor\frac{n}{15}\right\rfloor$ vertex-disjoint monochromatic triangles in our point set, scanning through the point set from left to right. These already cover $6\left\lfloor\frac{n}{15}\right\rfloor$ points. There are at most $9\left\lfloor\frac{n}{15}\right\rfloor+14$ points left. We can cover them in pairs using at most an amount of $\frac 9  2\left\lfloor\frac{n}{15}\right\rfloor+\frac{15}{2}$ additional triangles. This yields at most $\frac{13}{30}n+\frac{15}{2}$ triangles in total.
\end{proof}

\begin{restatable}{lemma}{lemgentriangles}
   \label{lem:gen triangles}
 	Given a bichromatic point set in the plane consisting of $b$ blue and $r$ red points in general position, there is a set of at most $\frac{2}{3}r+\frac{5}{3}$ disjoint triangles such that every blue point is contained in a triangle and no triangle contains a red point. 
\end{restatable}

\begin{figure}[hbt!]
	\centering
	\includegraphics[width=.5\linewidth]{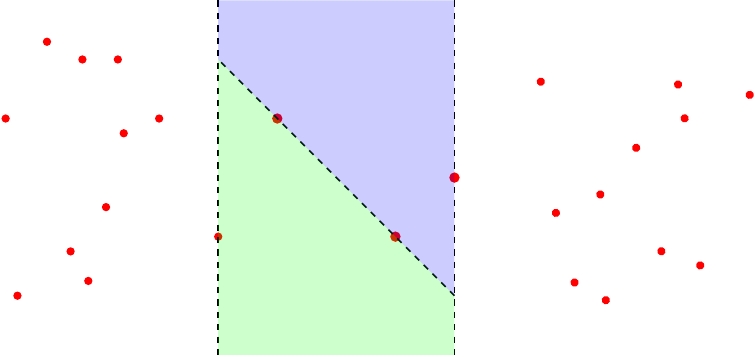}
	\caption{The two new generalized triangles for three new red points as in the proof of Lemma \ref{lem:gen triangles}.}
	\label{fig:gen triangles}
\end{figure}

\begin{proof}
    Let a \emph{generalized triangle} be an intersection of three half-planes such that every finite point set in it can be covered by a large enough triangle inside this intersection. 
	Note that it is sufficient to show that the plane can be covered by at most $\frac{2}{3}r+\frac{5}{3}$ generalized triangles such that all blue but none of the red points are contained in their disjoint interiors. 
 
	We assume without loss of generality that all points have pairwise different $x$-coordinates. Denote the $x$-coordinates of the red points by $x_1<\ldots<x_r$. Sweep a vertical line across the red points from left to right. First we cover the half-plane left of the vertical line $x=x_1$ by a single generalized triangle.
    At any given point we assume that the left half plane of the sweep line at $x=x_{r'}$ containing $r'$ red points is already fully covered by $\frac{2}{3}r'+\frac{1}{3}$ generalized triangles in the desired way. 
    Now consider the next three red points as in \Cref{fig:gen triangles}. Draw a vertical line $x=x_{r'+3}$. The two red points in between this line and the sweep line span a line splitting the strip between the vertical lines into two generalized triangles. Include these two generalized triangles into the covering and advance the sweeping vertical line. Note that the invariant is still fulfilled, since we increased $r'$ by 3 and used two additional generalized triangles.
    Repeat this until $r'>r-3$. If $r'=r$, we use the half-plane $x=x_r$ as the final generalized triangle to cover the plane with a total of $\frac{2}{3}r+\frac{4}{3}$ generalized triangles.
    Otherwise we add two generalized triangles for the right side of the sweeping line separated by a line through the remaining ($\leq 2$) red points.
    In total we cover the plane with at most $\frac{2}{3}(r-1)+\frac{1}{3}+2=\frac{2}{3}r+\frac{5}{3}$ generalized triangles since $r'\leq r-1$. This completes the proof.
\end{proof}

\begin{restatable}{proposition}{propbluedisjoint}
   Given a bichromatic planar $n$-point set, there are $\frac{2}{7}n+1$ disjoint triangles containing all blue points such that no triangle contains a red point.
   \label{prop:blue_disjoint}
\end{restatable}


\begin{proof}
    Let $b$ be the number of blue points and $r$ the number of red points. If $b\le\frac{4}{7}n+1$, then we can cover all blue points in pairs from left to right, such that the resulting triangles are disjoint. This way we need at most $\frac{b+1}{2}\le\frac{2}{7}n+1$ triangles. If $r\le\frac{3}{7}n-1$, we can apply Lemma \ref{lem:gen triangles} to cover all blue points using at most $\frac{2}{3}r\le\frac{2}{7}n+1$ triangles. Since $b+r=n$ one of these two has to be the case, finishing the proof.
\end{proof}

\begin{restatable}{lemma}{lemgreedy}
   Given a bichromatic point set in the plane consisting of $b$ blue and $b-3t$ red points in general position, there is a set of $t$ triangles with vertices at different blue points that do not contain a red point.
	\label{lem:b-3t}
\end{restatable}


\begin{proof}
	We apply induction on $t$. If $t=0$, we do not have anything to show. For the induction step let $t\geq 1$ and $p$ be an arbitrary blue point. 
	
	The other blue points have a fixed cyclic order around $p$ such that at most one angle between two consecutive ones of them is larger than $\pi$. Thus out of the $b-1$ angles, at least $b-2$ can be used to define blue triangles of $p$ and two consecutive points. The $b-2$ triangles constructed this way are interior disjoint. Since $t\geq1,$ the number of red points is at most $b-3$. Hence at least one of the $b-2$ triangles does not contain red points. We now apply induction on the $b-3$ blue vertices that are not incident to this triangle and all of the $b-3t=(b-3)-3(t-1)$ red points to obtain $t-1$ triangles with vertices at different blue points that do not contain a red point. Since the original triangle does not share a vertex with any of these by construction and does not contain red points either, the lemma follows.
\end{proof}

\begin{restatable}{corollary}{corgreedy}
\label{cor:greedily covering triangles}
   Given a bichromatic point set in the plane consisting of $b$ blue and $b-3t$ red points in general position, where $t\in\mathbb{N}_0, t\le \frac{b}{3}$, there is a set of $\lceil\frac{b-t}{2}\rceil$ triangles such that every blue point is contained in a triangle and no triangle contains a red point.
\end{restatable}


\begin{proof}
     By Lemma \ref{lem:b-3t}, there are $t$ vertex-disjoint blue triangles covering at least 3 blue points each and not containing any red point. The remaining $b-3t$ blue points can be covered in pairs as we did in \Cref{prop:match+comp}. Using Iverson brackets $[\cdot]$ the number of triangles used is
	\begin{equation*}
	t+\left\lceil\frac{b-3t}{2}\right\rceil = \frac{2t+b-3t}{2}+\frac{[2\nmid b-3t]}{2} = \frac{b-t}{2}+\frac{[{2\nmid b-t}]}{2} = \left\lceil\frac{b-t}{2}\right\rceil.
	\end{equation*}
\end{proof}

\begin{restatable}{proposition}{propcombinedupperbounds}
   \label{prop:combined upper bounds}
   Given a bichromatic planar $n$-point set, there is a set of $\frac{4}{15}n + \Oh(1)$ vertices, segments and triangles containing no red points and using all blue points as vertices exactly once.
\end{restatable}


\begin{proof}
	Let $r$ and $b$ be the number of red and blue points, respectively, and set $x:=\frac{b-r}{3}\leq \frac{b}{3}$. If there are more red points than blue points, then we can cover the blue points in pairs for an even better upper bound of $\frac{1}{2}b+\Oh(1)\leq\frac{1}{4}(r+b)+\Oh(1)$, so in the remaining case $x\geq 0$.
 Note that $n:=r+b=2b-3x$ is the total number of points. Taking the bounds from Lemma \ref{lem:gen triangles} with $t:=\lfloor x\rfloor$ and Corollary~\ref{cor:greedily covering triangles} with $r=b-3x$, we get that the number of triangles required per point is upper bounded by
	\begin{equation*}
	\max_{x\in[0,\frac{b}{3}]}\min\left(\frac{\frac{2}{3}(b-3x)+\frac{5}{3}}{n},\frac{1}{n}\left\lceil\frac{b-\lfloor x\rfloor}{2}\right\rceil\right).
	\end{equation*}
 This expression can be upper bounded by
 \begin{equation*}
    \max_{x\in[0,\frac{b}{3}]}\min\left(\frac{2x-\frac{2}{3}b-\frac{5}{3}}{3x-2b}, \frac{x-b-2}{6x-4b}\right)
 \end{equation*}
 via rearranging the terms and observing that
 \begin{equation*}
    \left\lceil\frac{b-\lfloor x\rfloor}{2}\right\rceil\le\frac{b-\lfloor x\rfloor+1}{2}\le\frac{b- x+2}{2}.
 \end{equation*}
 Regarding $b$ as a constant, we can see that both expressions in the minimum are monotone in $x$ for $x\in[0,\frac{b}{3}]$. The first is monotonically decreasing while the second is monotonically increasing. This implies that the worst case is attained for the unique value of $x$ making both expressions equal. Setting them equal yields $x=\frac{1}{9}(b-16)$. Therefore the number of triangles required per point is upper bounded by
 \begin{equation*}
    \frac{\frac{1}{9}(b-16)-b-2}{\frac{6}{9}(b-16)-4b+4} = \frac{4b+17}{15b+30}.
 \end{equation*}
 This completes the proof.
\end{proof}

\begin{restatable}{proposition}{propnobluetriangle}
   \label{prop:no blue triangles}
   There are bichromatic planar $n$-point sets such that there is no set of $\lfloor\frac{n}{4}\rfloor$ triangles containing all blue points and no set of $\frac{3}{8}n+\Oh(1)$ triangles containing all points such that no triangle contains points of both colors.
\end{restatable}


\begin{proof}
	Consider the following construction. Place blue points $p_1,\ldots, p_b$ on a circle in that order and draw the complete graph on them with straight lines. Let $e$ be the edge between $p_b$ and $p_1$. Each point $p_i$ is incident to $b-2$ cells. If $i\neq 1,b$ exactly one of these cells is inside the triangle defined by $e$ and $p_i$. Place a red point $r_i$ inside each such cell. This way we place $b-2$ red points. See \Cref{fig:no blue triangle} for illustration.
	
\begin{figure}[htb]
\begin{minipage}{.33\textwidth}
\centering
\includegraphics[width=.9\linewidth]{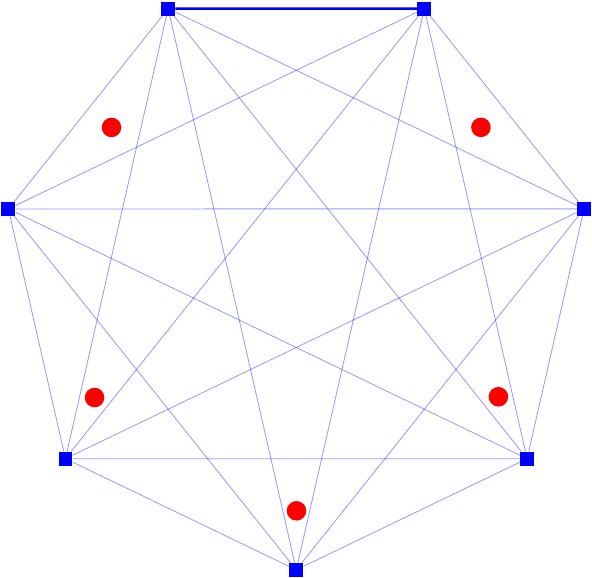}
\end{minipage}%
\begin{minipage}{.33\textwidth}
\centering
\includegraphics[width=.9\linewidth]{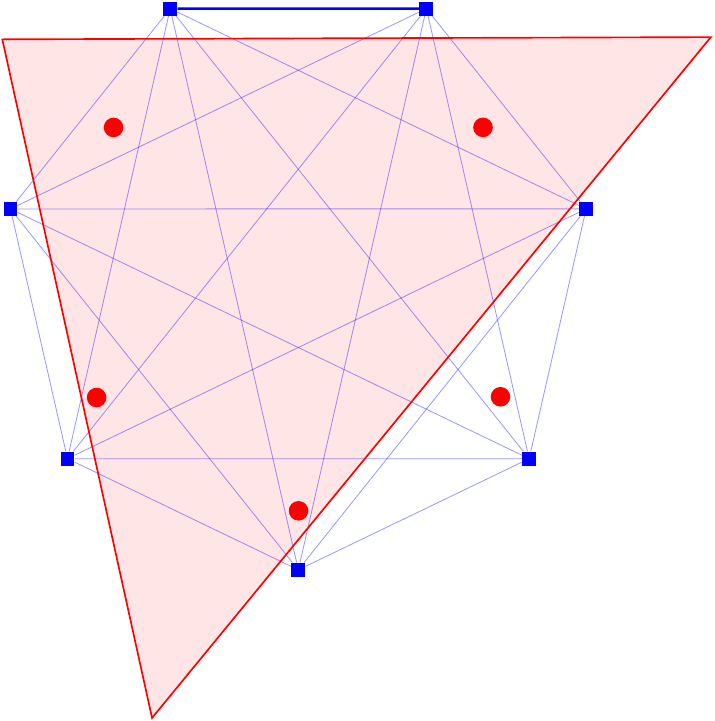}
\end{minipage}%
\begin{minipage}{.33\textwidth}
\centering
\includegraphics[width=.9\linewidth]{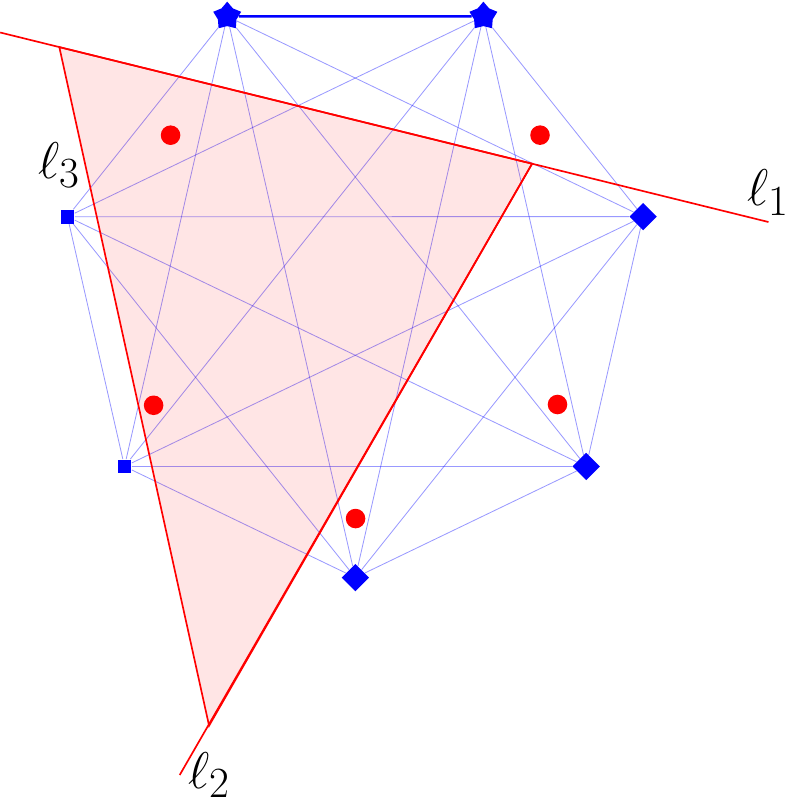}
\end{minipage}%
	\caption{Illustrations for the proof of \Cref{prop:no blue triangles}. From left to right, (a): The construction. (b): A red triangle covering 4 red points. (c): The lines $\ell_1,\ell_2$ and $\ell_3$ as in the proof. $B_1$ consists of the star-shaped blue points, $B_2$ of the diamond-shaped blue points and $B_3$ of the square-shaped blue points.}
	\label{fig:no blue triangle}
\end{figure}
	\begin{claim*}
	    There is no blue triangle containing three blue points. 
	\end{claim*} 
    \begin{claimproof}
        We see that it is enough to consider triangles with blue points as corners as any triangle containing three blue points contains one of these in its interior. Let $abc$ be a triangle where $a,b$ and $c$ are blue points. Let $a$ and $c$ be the points to the left and to the right of the edge $e$ along the circle. We see that $b$ is not a point of $e$, i.e. we have placed a red point $p$ next to $b$. The line $bp$ intersects $e$ and hence also intersects the line $ac$. Therefore $p$ is inside $abc$.
    \end{claimproof}

	\begin{claim*}
		There is no red triangle containing more than four red points.
	\end{claim*}

 \begin{claimproof}
     Consider the three lines $\ell_1,\ell_2$ and $\ell_3$ defined by the edges of a red triangle. For every blue point at least one of these lines separates the triangle and this point. Let $B_1$ be the set of blue points separated by $\ell_1$ and from among the rest of the points the points of $B_2$ are separated by $\ell_2$. Let $B_3$ be the remaining points which are thus separated from the triangle by $\ell_3$. 
	In particular, $|B_1|+|B_2|+|B_3|= b$. Note that since we cut them off using a line, all three of the $B_i$ are consecutive points along the convex hull of blue points, so their convex hulls do not intersect each other or the triangle. But to avoid any blue triangles in $B_i$ we need at least $|B_i|-2$ red points inside of its convex hull, see Lemma \ref{lem:b-3t}.
	
	
	Hence there are at least $|B_1|-2+|B_2|-2+|B_3|-2=b-6$ red points that are not contained in the triangle. This implies that at most 4 red points are covered. It is easy to see that indeed 4 red points can be covered by a triangle not containing a blue point. 
 \end{claimproof}

 \medskip
	
	With these two claims it is now clear that we need at least $\frac{1}{4} (b-2)$ triangles for the red points and at least $\frac{1}{2} b$ triangles for the blue points. As $n=2b-2$ we need $k+1$ triangles for the blue points if and only if $b\geq 2k+1\Leftrightarrow {n\geq 4k }$ and up to the constant $\frac{3}{4}b\approx\frac{3}{8}n$ triangles in total.
\end{proof}

\section{Computational Aspects}
\label{sec:computational}

For a point set in general position which we study in this paper, three points $a,b,c$ are either oriented clockwise or counterclockwise. 
This gives rise to so-called triple orientations, which can be encoded as Boolean variables.
Using these variables we encode questions about sets of a fixed number of points 
as a Boolean satisfiability instance,
which we then study using a SAT solver.

Not every assignment of triple orientations can be represented by a point set
and deciding realizability is in fact a computationally hard problem \cite{mnev1988universality}.
We can however add clauses to ensure that certain necessary conditions are fulfilled.
More specifically, we consider the search space of so-called \emph{pseudoconfiguration of points}, which contain all point sets.
Besides using the combinatorial structure of triangle orientations, we assume without loss of generality that the points have distinct $x$-coordinate, which allows us to order them from left to right with increasing $x$-coordinate.
The triangle orientation of those points are encoded in the combinatorial setting of rank~3 signotopes. This structure turns out to be more efficient in the encoding since we have to add less constraints.

We use the encoding from~\cite{Scheucher2020}, which comes with one Boolean variable 
for each triple of points 
to indicate whether the triple is positively or negatively oriented. Clauses ensure that solutions of the Boolean formula correspond to signotopes.
In addition, we need to assign a color to each point. For this we introduce variables and clauses which ensure that each point has exactly one color. 

For the considered problem, we look for triangles covering some points of the same color and not containing points of the other color. 
In order to formulate it as a CNF, we only consider triangles which are spanned by points of the actual point set. Three points give a proper triangle. However, if only two points are contained in a triangle we always find a sufficiently small triangle covering the straight-line segment between those two points without containing any other points of the point set. Similarly this holds for triangles consisting of only one point. 
Moreover,
the encoding comes with auxiliary variables to indicate whether a pair of edges cross,
and auxiliary variables to indicate whether a triangle contains points in its interior.
Since those properties are characterized by the triple orientations, these auxiliary variables are synchronized with the orientation variables via clauses.

\subsection{Full cover}

In the proof of \cref{prop:induct+comp} we study a bichromatic point set $P$ and consider the problem whether we can cover this point set with $k$ monochromatic triangles such that each triangle is spanned by at most 3 points of $P$ with the same color and does not contain any points of the other color. 
To simplify this problem with computer investigation, we consider triangles which do not contain any other points of $P$.

Assume towards a contradiction that a bichromatic point set $P=\{p_1,\ldots,p_n\}$ exists 
which cannot be covered by $k$ triangles.
We can assume without loss of generality that the points $p_1,\ldots,p_n$ have increasing $x$-coordinate,
and therefore
the induced triple-orientations form a rank~3 signotope on~$[n]$ as discussed above. 

To the general setting above which encodes bichromatic point sets, we add clauses to ensure that no partition  
$[n] = \bigcup_{i=1}^k I_j$ 
with $|I_j| \in \{1,2,3\}$ is valid.
In particular, for each such partition  $(I_1,\ldots,I_k)$,
we add a clause to ensure that there is a pair $I_a,I_b$ which is not disjoint
or one $I_a$ does not span a monochromatic empty triangle. 
If two $I_a,I_b$ parts are not disjoint, then their convex hulls have an intersection, which we can encode using the intersection auxiliary variables. 
In the other case if one $I_a$ does not span a monochromatic empty triangle, then either two points of $I_a$ have the same color or the triangle contains a point of $[n]$. Those conditions can be encoded using the variables encoding the colors of the points or the empty triangle auxiliary variables. 
Even though the number of partitions $[n] = \bigcup_{i=1}^k I_j$  grows exponentially in $n$,
the instances for up to $n=15$ are of reasonable size and can be performed on a laptop. For larger~$n$, where more RAM was required, we used a computing cluster.

Moreover, we used the symmetry breaking (as explained in~\cite{Scheucher2020}) to reduce the search space: without loss of generality the points $s_2,\ldots,s_n$ appear in this cyclic order around~$s_1$ (which lies on the convex hull).

\subsection{Two vertex-disjoint triangles}

In~\cref{prop:match+comp}, we want to find a lower bound on the number of triangles with which every bichromatic point set in general position in the plane can be covered. 
As shown in \cite{Scheucher2020}, every set of 17 points
contains a pair of area-disjoint empty pentagons \cite{Scheucher2020}. 
By pigeon hole principle every such empty pentagon contains a monochromatic empty triangle. 
Hence 17 is an upper bound on the number of points required to find two vertex-disjoint monochromatic empty triangles.
Partitioning the point set with a sweeping line in parts of 17 points shows that 6 out of those 17 points can be covered with two triangles. The remaining points can be covered with monochromatic matching edges. 

Using the SAT framework, we show that every set of 15 points contains a pair of vertex-disjoint triangles, which clearly improves the bound as discussed in the proof of~\cref{prop:match+comp}.
Since there exists a set of 14 points with no pair of vertex-disjoint triangles, 
the bound is optimal.

For this problem we look for a bichromatic point set which avoids two vertex-disjoint triangles spanned by exactly three points of the same color. 
Again we start with the framework explained above and model this problem by adding a clause that for each pair of monochromatic triangles they intersect or one of them is not empty. 
Again this can be done with the auxiliary variables explained above. 

Additionally we checked for further improvement, by allowing that the vertex disjoint monochromatic triangles contain points of the same color. Using our SAT framework this does not yield an improvement as there are bichromatic point sets with 14 points without two such monochromatic triangles but each set of 15 points contains two such triangles.  
In order to get an improvement we also tried to find three vertex disjoint monchromatic triangles. 
Using our program we found a pseudopoint configuration on 20 elements with no three disjoint triangles.
The computations took about 10 CPU hours. Since we add clauses for each set of three triangles, the instances grow fast and we did not manage to start larger instances.

\end{document}